\documentclass[a4paper,11pt]{article}
\usepackage{setspace}
\usepackage[linktoc=page,colorlinks=true,linkcolor=blue,citecolor=blue]{hyperref}
\usepackage{latexsym,amssymb,amsmath,mathrsfs,amsthm,rotating}
\usepackage{fullpage}
\usepackage{natbib}
\usepackage[retainorgcmds]{IEEEtrantools}
\usepackage{wrapfig,color,dashrule,epsfig,graphics,epsfig,graphicx,multirow}
\usepackage{pstool,psfrag}
\usepackage{algorithmicx}
\usepackage{algorithm}
\usepackage{algpseudocode}
\usepackage{authblk}
\makeatletter

\newcommand{\Rmnum}[1]{\expandafter\@slowromancap\romannumeral #1@}
\newcommand{\mbbR}{\mbox{$\mathbb{R}$}}

\newcommand{\mbbP}{\mbox{$\mathbb{P}$}}

\makeatother
\date{}
\title{Stochastic Ordering under Conditional Modelling of Extreme
  Values: Drug-Induced Liver Injury}

\author[*]{Ioannis Papastathopoulos}
\author[$\dagger$]{Jonathan A. Tawn}
\affil[*]{\small School of Mathematics, University of Bristol}
\affil[$\dagger$]{\small Department of Mathematics and Statistics, Lancaster University}
\affil[$$]{\small i.papastathopoulos@bristol.ac.uk $\quad$ j.tawn@lancaster.ac.uk}

\newtheorem{theorem}{Theorem}

\newenvironment{proof}{\begin{trivlist} \item[]
    {\bf Proof.}}{\nolinebreak \hfill \rule{2mm}{2mm} \end{trivlist}}

\begin{document}
\maketitle
%\tableofcontents
%\listoffigures

\begin{abstract}
  Drug-induced liver injury (DILI) is a major public health issue and
  of serious concern for the pharmaceutical industry.\ Early detection
  of signs of a drug's potential for DILI is vital for pharmaceutical
  companies' evaluation of new drugs.\ A combination of extreme values
  of liver specific variables indicate potential DILI (Hy's Law).\ We
  estimate the probability of joint extreme elevations of laboratory
  variables using the conditional approach to multivariate extremes
  which concerns the distribution of a random vector given an extreme
  component.\ We extend the current model to include the assumption of
  stochastically ordered survival curves and construct a hypothesis
  test for ordered tail dependence between doses, a pattern that is
  potentially triggered by DILI. The model proposed is applied to
  safety data from a Phase 3 clinical trial of a drug that has been
  linked with liver toxicity.
\end{abstract}
%\keywords{Stochastic Ordering, Multivariate Extremes}

% -------------------------------------------------------------------------
\section{Introduction}
\label{sec:intro}
% -------------------------------------------------------------------------
% 
% -------------------------------------------------------------------------
Drug-induced liver injury (DILI) is a major public health and
industrial issue that has concerned clinicians for the past 50 years.\
\cite{FDAliver} reports that many drugs for a diverse range of
diseases were either removed from the market or rejected at the
pre-marketing stage because of severe DILI (e.g.,\ iproniazid,
ticrynafen, benoxaprofen, bromfenac, troglitazone, nefazodone, etc.).\
Therefore, signals of a drug's potential for DILI and early detection
can help to improve the evaluation of drugs and aid pharmaceutical
companies in their decision making.\ However, in most clinical trials
of hepatotoxic drugs, evidence of hepatotoxicity is very rare and
although the pattern of injury can vary, there are no pathognomonic
findings that make diagnosis of DILI certain, even upon liver biopsy.\
Indeed, most of the drugs withdrawn from the market for
hepatotoxicity, fall mainly in the post-marketing category, and have
caused death or transplantation at frequencies of less than 1 per
10000 people that have been administered the drug.

Although the mechanism that causes DILI is not fully understood yet,
the procedure under which its clinical assessment is performed stems
from Zimmerman's observation that hepatocellular injury sufficient to
impair bilirubin excretion is a revealing indicator of DILI (Zimmerman
1978, 1999), also informally known as Hy's Law.\ In other words, a
finding of \textit{alanine aminotransferase} (ALT) elevation, usually
substantial and greater than thrice the upper limit of the normal of
ALT (ULN$_{\text{ALT}}$), seen concurrently with \textit{bilirubin}
(TBL) greater than twice the upper limit of the normal of TBL
(ULN$_{\text{TBL}}$), identifies a drug likely to cause severe DILI
(fatal or requiring transplant).\ Moreover, these elevations should
not be attributed to any other cause of injury, such as other drugs,
and \textit{alkaline phosphatase} (ALP) should not be greatly elevated
so as to explain TBL's elevation.

\cite{southheff12b} identified the assessment of DILI as a
multivariate extreme value problem and using the \cite{hefftawn04}
modelling approach, analysed liver-related laboratory data.\ The use
of the \cite{hefftawn04} model in this context is supported by the
flexibility of the model to allow a broad class of dependence
structures and the possibility to describe the probabilistic behaviour
of a random vector which is extreme in at least one margin.\ Despite
its strong modelling potential, complications in terms of parameter
identifiability problems and invalid inferences are experienced with
the original modelling procedure of \cite{hefftawn04}.\
\cite{keefpaptawn11} provided missing constraints for the parameter
space of the \cite{hefftawn04} model that are aimed to overcome these
complications.

The data we consider in this study relates to observed liver-related
variables from a sample of 606 patients who were issued a drug that
has been linked to liver injury in a Phase 3 clinical trial and can be
found in \cite{texmex}; see also \cite{southheff12b}.\ The patients
were categorised into 4 different dose levels in a randomised,
parallel group, double blind Phase 3 clinical study.\ Our main
question in this paper about the data is whether they support evidence
of toxicity with increasing dose.\ This signal would be justified by a
significant positive probability of post-baseline ALT and TBL being
greater than $3 \times \text{ULN}_{\text{ALT}}$ and $2 \times
\text{ULN}_{\text{TBL}}$, respectively.\ However, insufficient trial
duration and the small sample sizes encountered in most such
applications may lead to estimated zero probabilities of DILI for all
doses.\ This would stem from the non-occurrence of joint ALT and TBL
elevations or from inaccurate extrapolation due to the limited source
of information.\ Therefore, other patterns that could indicate or be
triggered by DILI would be helpful and here we consider an alternative
approach for assessing evidence of altered liver behaviour.

The current understanding of the biology that underpins Hy's law, is
that liver cells leak ALT into the blood as they are damaged. As the
amount of damage increases, the amount of ALT increases, and so the
liver begins to lose its capacity to clear TBL. Subsequently, TBL is
also expected to start to increase. At levels of damage that do not
affect liver's ability to clear TBL, dependence is not
expected. Hence, given that the drug has increasing toxicity with
dose, we expect a natural ordering in the joint tail area of ALT and
TBL.\ This pattern of tail ordering is the main focus of this paper
and is used to aid inference as well as to improve estimation
efficiency in the modelling procedure of DILI.

\begin{table}[htbp!]
  \caption{Conditional Spearman's correlation estimates between ALT and TBL for four different dose levels and two conditioning levels, i.e.\, 100\% and 20\%. The letters $A$, $B$, $C$ and $D$ 
    represent, in increasing order, the amount of the dose.}    \vspace{5pt}
  \centering
  \begin{tabular}{r | r r | r r | r r}    
    \hline
    &\multicolumn{2}{c|}{Baseline} & \multicolumn{2}{c|}{Post-baseline} &
    \multicolumn{2}{c}{Residual} \\
    \hline
    & 100\% & 20\% & 100\% & 20\% & 100\% & 20\% \\
    \hline
    Dose A& 0.15 &  0.15 & 0.32 &  0.03  & 0.14  &  0.00 \\
    Dose B& 0.10 & -0.17 & 0.15 & -0.17  & -0.06 & -0.19\\
    Dose C& 0.12 & -0.11 & 0.11 & -0.05  & -0.08 & -0.14\\
    Dose D& 0.17 & -0.03 & 0.23 &  0.02  & 0.05  & -0.15
  \end{tabular}
  \label{table:kendall}
\end{table}

Our approach consists of developing, similar to
\cite{keefpaptawn11}, constraints that describe ordered dependence in
the joint tail. These new constraints can be used to test, via a
likelihood ratio test procedure, the effect of stochastic ordering in
the tails of ALT and TBL across dose and hence provide a signal of
altered liver behaviour. Subsequently, inference can be sharpened by
incorporating scientific knowledge in the modelling process through
the imposition of the ordering constraints. Our motivation to impose
ordering constraints in the tail stems from the fact that the
probability of DILI is logically ordered between different dose levels
when the drug is liver toxic and this feature is unlikely to be
evident from the data when considering clinical trials with small
sample sizes as this particular one. Therefore, estimation of ALT and
TBL under the dose ordering assumption is beneficial as it removes
variability that arises in small sample sizes in the joint tail region
of ALT and TBL.

For example, the two central columns of Table~\ref{table:kendall} show
the estimated conditional Spearman's rank correlation \citep{Schmid06}
of ALT and TBL measured at the baseline and post-baseline periods.\
The value $100\%$ corresponds to the usual Spearman's correlation
whereas the value $20\%$ corresponds to the rank correlation in the
upper $[0.8,1]\times [0.8,1]$ tail region of the copula space
\citep{genenesl12} of the bivariate random variable
$(\text{ALT},\text{TBL})$.\ Similarly, the last column shows the
conditional Spearman's correlation of the residuals of ALT and TBL
which are the log-measurements corrected for baseline differences, see
also Section~\ref{sec:Introapp}.\ Table~\ref{table:kendall} conveys
this lack of ordering since the estimated dependence of ALT and TBL at
dose A appears to be higher than any other dose in post-baseline and
residual scale.\ On the other hand, the published literature reports
jaundice, hepatitis and similar symptoms in approximately 1 out of 500
patients taking the dose D of this drug \citep{southheff12b}.\ We view
such high dependencies in low doses as a by-product of sampling
variability.

The proposed methodology and data analysis of the paper are based on
an asymptotically motivated model of multivariate extreme value
threshold model which is fitted to a fraction of the data.\ An
alternative approach would be to model the joint distribution between
the variables using all the data through the use of empirically
selected marginal and copula models \citep{joe97,Nelsen06,genenesl12}.\ The
former approach should exhibit less bias but larger variability than
the latter, thus as the sample size increases the extreme value
approach is likely to become more the efficient.\ The sample size in
our study is probably about at the boundary where the extreme value
methods have an advantage.\ Furthermore, even if the copula approach
were to be adopted, the strategies developed here would still be
relevant as the stochastic ordering issue would need addressing.

The paper is organised as follows.\ Section~\ref{sec:method} describes
the conditional dependence model of \cite{hefftawn04} and the
constraints of \cite{keefpaptawn11}.\ The additional constraints based
on the assumption of stochastically ordered conditional distributions
are presented, together with a likelihood ratio test of tail ordering,
in Section~\ref{sec:stochorder}.\ The effect of the constraints of
\cite{keefpaptawn11} and this paper are assessed with a simulation
study in Section~\ref{sec:simulation}.\ In
Section~\ref{sec:application} we apply the additional constraints in
an analysis of multivariate extremes of ALT and TBL of the DILI data.

%-------------------------------------------------------------------------
\section{Methodology}
\label{sec:method}
%-------------------------------------------------------------------------

%-------------------------------------------------------------------------
\subsection{Marginal transformation}
\label{sec:Laptrans}
%-------------------------------------------------------------------------

Here and throughout vector algebra is applied componentwise.\ Let
$\boldsymbol{X}=(X_{1},...,X_{d})$ be a continuous $d$-dimensional
random vector and $\Delta=\{1,...,d\}$.\ We adopt the
marginal transformation to approximate Laplace margins
\citep{keefpaptawn11}
\begin{IEEEeqnarray}{rCl}
  Y_{i}=\left\{ \begin{array}{cl}
      \log\{2\hat{F}_{X_i}\left(X_{i}\right)\} &\text{for
        $\hat{F}_{X_i}(X_{i})<\frac{1}{2}$,}   \\
      -\log\{2[1-\hat{F}_{X_i}\left(X_i\right)]\} &\text{for
        $\hat{F}_{X_i}(X_{i})\geq \frac{1}{2}$,}
    \end{array}\right.
  \label{eq:laplace}
\end{IEEEeqnarray}
where the estimated distribution function $\hat{F}_{X_{i}}$ is
obtained from the semi-parametric model of \cite{coletawn94}
\begin{IEEEeqnarray}{rCl}
  \hat{F}_{X_{i}}(x) = 
  \left\{
    \begin{array}{ll}
      1-\{1-\tilde{F}_{X_{i}}(u_{X_{i}})\}\{1+\hat{\xi}_{i}(x-u_{X_{i}})/\hat{\sigma}_{i}\}_
      {+}^
      {-1/\hat{\xi}_{i}}
      & \mbox{for $x>u_{X_{i}}$,}      
      \\
      \tilde{F}_{X_{i}}(x) & \mbox{for $x\leq u_{X_{i}}$.}
    \end{array}
    \right.      
    \label{eq:coletawn}
\end{IEEEeqnarray}
Here, $\tilde{F}_{X_{i}}$ is the empirical distribution function and
$u_{X_{i}}$ is a threshold above which the generalised Pareto
distribution with scale parameter and shape parameters $\sigma_i$ and
$\xi_i$, short-hand GP($\sigma_{i},\xi_{i}$), $\sigma_{i}\in
(0,\infty)$, $\xi_{i}\in \mbbR$, is fitted to the observed values of
the excess random variable $X_{i}-u_{X_{i}}|X_{i}>u_{X_{i}}$
\citep{davismit90}.\ The choice of the transformation to Laplace
marginals is motivated by the symmetry of the Laplace distribution
that ensures the model is unchanged for negatively dependent variables
\citep{keefpaptawn11}.

% \cite{nadaande98} give the exact conditions under which two
% generalised extreme value distributions are stochastically
% ordered. Similarly, for GPD($\sigma_{i},\xi_{i}$) distributions
% $F_{X_{1}}(x),\hdots, F_{X_{d}}(x)$, the ordering constraint
% $F_{X_{1}}(x) \leq \hdots \leq F_{X_{d}}(x)$, is satisfied in the
% simplest case where the shape parameters are identical,
% i.e. $\xi_{1}= \hdots =\xi_{d}$, when the following constraint is
% true
% \begin{equation} \sigma_{1} \geq \hdots \geq \sigma_{d}.
% \end{equation}

%\newpage
%-------------------------------------------------------------------------
\subsection{Conditional modelling of extreme values}
\label{sec:condmodel}
%-------------------------------------------------------------------------

The \cite{hefftawn04} conditional dependence model characterises the
probabilistic behaviour of the conditional random vector
$\boldsymbol{Y}_{-i}|Y_{i}=y$, for large $y$.\ The random vector
$\boldsymbol{Y}_{-i}$ denotes the $(d-1)$-dimensional vector of the
transformed variables without the $i$-th margin.\ According to
\cite{hefftawn04}, for each $i \in \Delta$, there exist vector-valued
normalising functions, $\boldsymbol{a}_{|i}: \mathbb{R} \rightarrow
\mathbb{R}^{d-1}$ and $\boldsymbol{b}_{|i}: \mathbb{R} \rightarrow
\mathbb{R}^{d-1}$, such that for $x>0$
\begin{IEEEeqnarray}{rCl}
  \mbbP\left\{Y_{i}-u>x,\frac{\boldsymbol{Y}_{-i}-\boldsymbol{a}_{|i}
      \left(Y_{i}\right)}
    {\boldsymbol{b}_{|i}\left(Y_{i}\right)}\leq\boldsymbol{z}\bigg |
    Y_{i}>u \right\} & \rightarrow &
  \exp\left(-x\right)G_{|i}\left(\boldsymbol{z}\right) \quad\text{as
    $u \rightarrow \infty$},
  \label{eq:G}
\end{IEEEeqnarray}
where the $j$th marginal distribution $G_{j|i}$ of $G_{|i}$ is a
non-degenerate distribution function for all $j \in
\Delta\setminus\{i\}$ and additionally, the following condition is
required such that $G_{|i}$ is uniquely-defined
\begin{IEEEeqnarray}{rCl}
  \lim_{z\rightarrow \infty}G_{j|i}\left(z\right) & = & 1 \quad
  \text{for all $j \neq i$}, \nonumber
\end{IEEEeqnarray}
so there is no mass at $+\infty$ but some is allowed at $-\infty$, in
any margin.\ \cite{hefftawn04} identified that the normalising
functions are unique up to type, and for a broad class of
distributions, \cite{keefpaptawn11} showed that these functions are
all in the parametric family
\begin{IEEEeqnarray}{rClCl}
  \boldsymbol{a}_{|i}\left(x\right) &=& \boldsymbol{\alpha}_{|i}x
  \quad \text{and} \quad \boldsymbol{b}_{|i}\left(x\right) &=&
  x^{\boldsymbol{\beta}_{|i}},\nonumber
\end{IEEEeqnarray}
with $(\boldsymbol{\alpha}_{|i},\boldsymbol{\beta}_{|i}) \in
\left[-1,1\right]^{d-1}\times\left(-\infty,1\right)^{d-1}$ and $x>0$.\
Positive and negative dependence between variables $Y_{i},Y_{j}$, for
$i \neq j$ is given by $\alpha_{j|i}>0$ and $\alpha_{j|i}<0$,
respectively, with $\alpha_{j|i}$ the associated
$\boldsymbol{\alpha}_{|i}$ with $Y_{j}$ variable.\ The strongest form
of positive (negative) extremal dependence occurs when
$\alpha_{j|i}=1$ ($\alpha_{j|i}=-1$) and $\beta_{j|i}=0$ and is termed
as asymptotic positive (negative) dependence, for all $j \neq i$.\
Otherwise, variables are termed asymptotically independent.\ 

The conditional model of \cite{hefftawn04} can be viewed as a
multivariate semiparametric regression of $\boldsymbol{Y}_{-i}$ on $
Y_{i}$, i.e.\, given $Y_{i}>u$, for large $u$
\begin{IEEEeqnarray}{rCl}
  \boldsymbol{Y}_{-i} & = & \boldsymbol{\alpha}_{|i}x +
  x^{\boldsymbol{\beta}_{|i}}\boldsymbol{Z}_{|i} \quad \text{for
    $Y_{i}=x>u$}
\end{IEEEeqnarray}
where $\boldsymbol{Z}_{|i}$ is a $d-1$ dimensional variable with
non-zero mean and distribution function $G_{|i}$.\ The original
procedure of \cite{hefftawn04} for estimating the vector parameters
$\boldsymbol{\alpha}_{|i}$ and $\boldsymbol{\beta}_{|i}$ consists of
using pseudo-likelihood methods to jointly estimate the parameters of
interest.\ In particular, if $\boldsymbol{Z}_{|i}$ has finite vector
mean $\boldsymbol{\mu}_{|i}$ and standard deviations
$\boldsymbol{\sigma}_{|i}$, then the mean and standard deviation of
the conditional random variable $\boldsymbol{Y}_{-i}|Y_{i}>u$ is
$\boldsymbol{\alpha}_{|i}Y_{i}+
\boldsymbol{\mu}_{|i}\left(Y_{i}\right)^{\boldsymbol{\beta}_{|i}}$ and
$Y_{i}^{\boldsymbol{\beta}_{|i}}\boldsymbol{\sigma}_{|i}$,
respectively.\ Under the false working assumption that
$\boldsymbol{Z}_{|i}$ are independent Normal random variables,
numerical maximisation of the likelihood over the parameter space is
required to obtain parameter estimates
$(\hat{\boldsymbol{\alpha}}_{|i},
\hat{\boldsymbol{\beta}}_{|i},\hat{\boldsymbol{\mu}}_{|i},
\hat{\boldsymbol{\sigma}}_{|i})$, and $G_{|i}$ is estimated
nonparametrically by the empirical distribution function of:
\begin{equation}
  \hat{\boldsymbol{Z}}_{|i} = \frac{\boldsymbol{Y}_{-i} -
    \hat{\boldsymbol{\alpha}}_{|i}Y_{i}  }
  {(Y_{i})^{\hat{\boldsymbol{\beta}}_{|i}}}. 
\end{equation}

Given parameter estimates, standard procedures for inference and
extrapolation can be performed as in \cite{hefftawn04} by implementing
Algorithm \ref{alg:1}.\ As an example, the functional
$\mbbP\left(\boldsymbol{X} \in C|X_{i}>s\right)$ can be approximated
by repeating steps 1--5, and evaluating the estimate as the long run
proportion of the generated sample that falls in a set $C \in
\mathbb{R}^{d}$.\ As far as the confidence intervals of the estimate
of any functional are concerned, these are obtained by the replication
of the three stages of the following bootstrap method: data generation
under the fitted model, estimation of model parameters and the
derivation of an estimate of any derived parameters linked to
extrapolation.
\begin{algorithm}
  \caption{Sampling Algorithm}
  \label{alg:1}
  \begin{algorithmic}[1] \State Simulate $Y_{i}$ from the Laplace
    distribution conditional on its exceeding threshold $u>0$.
    
    \vspace{5pt}
    
\State Sample $\boldsymbol{Z}_{|i}$, independently of $Y_{i}$, from
the empirical distribution function, $\hat{G}_{|i}$.
    
    \vspace{5pt}
    
    \State Obtain $\boldsymbol{Y}_{-i}=\boldsymbol{\alpha}_{|i}Y_{i} +
    (Y_{i})^{\boldsymbol{\beta}_{|i}} \boldsymbol{Z}_{|i}$.
    
    \vspace{5pt}
    
    \State Transform $\boldsymbol{Y}=(\boldsymbol{Y}_{-i},Y_{i})$ to the
    original scale by using the inverse transformation of equation
    (\ref{eq:laplace}) for each margin.
  
    \vspace{5pt}
  
    \State The resulting transformed vector $\boldsymbol{X}$ constitutes
    a simulated value from the conditional distribution of
    $\boldsymbol{X}|X_{i}>t^{-1}(u)$, where $t^{-1}(\cdot)$ denotes the
    inverse transformation of equation (\ref{eq:laplace}).
  \end{algorithmic}
\end{algorithm}

% -------------------------------------------------------------------------
\subsection{Inference based on \cite{keefpaptawn11}}
\label{sec:inference}
%-------------------------------------------------------------------------
Although the efficiency of the model has led to its implementation in
a wide range of applications including riverflow and rainfall
\citep{keefST09}, temporal river flow cases \citep{eastoetawn12}, food
safety \citep{pauletal06} and finance \citep{abbas11}, it was recently
discovered by \cite{keefpaptawn11} that further constraints on the
parameter space of the model are required.\ According to Sibuya (1960)
and Tiago de Oliveira (1962/63), there are different categorisations
of extremal dependence between two random variables $(X_{i},X_{j})$,
i.e.\, asymptotic dependence and asymptotic independence measured by
the coefficients of tail dependence
\begin{IEEEeqnarray}{rCl}
  \chi_{ij}^{+} & = & \lim_{p\rightarrow 1}\mbbP\left\{X_{j} >
    F_{j}^{-1}(p)|X_{i} > F_{i}^{-1}(p)\right\},\nonumber\\\label{eq:chi}\\
  \chi_{ij}^{-} & = & \lim_{p\rightarrow 1}
  \mbbP\left\{X_{j} <
    F_{j}^{-1}(1-p)|X_{i} > F_{i}^{-1}(p)\right\}.\nonumber
\end{IEEEeqnarray}
When $\chi_{ij}^{+}>0$ ($\chi_{ij}^{-}>0$) the variables are termed
asymptotically positive (negative) dependent and asymptotically
independent, otherwise.\ Taking these measures into consideration,
\cite{hefftawn04} omitted the fact that there is stochastic ordering
between asymptotically independent and dependent models. % since
% $\chi_{ij}^{+}>0$ and $ \chi_{ij}^{-} >0$ holds for any type of
% extremal dependence.  , where $\chi_{ij}^{+}, \chi_{ij}^{-},
% \chi_{ij}$ denote the coefficients of tail dependence of
% asymptotically positive dependent, asymptotically negative dependent
% and asymptotically independent variables, respectively.\
In particular, let the $q$th conditional quantile of $Y_{j}|Y_{i}=x$,
for large $x$ under the \cite{hefftawn04} model be
$y_{j|i}(q)=\alpha_{j|i}x+x^{\beta_{j|i}}z_{j|i}(q)$, the associated
quantile under asymptotic positive dependence
$y_{j|i}^{+}(q)=x+z_{j|i}^{+}(q)$ and the associated quantile under
asymptotic negative dependence $y_{j|i}^{-}(q)=-x+z_{j|i}^{-}(q)$.\
The natural restriction
\begin{equation}
  y_{j|i}^{-}(q) \leq y_{j|i}(q) \leq y_{j|i}^{+}(q), \quad \text{for all } q\in[0,1],
  \label{eq:keef_constr}
\end{equation}
imposes further constraints on the parameter space of the model which
are given by (Theorem 1.1,~\cite{keefpaptawn11}), i.e.\, for all
$q\in [0,1]$:\newline

\noindent Case I:  either
\[
\alpha_{j|i} \leq \min\left\{1, 1 - \beta_{j|i} z_{j|i}(q)
  v^{\beta_{j|i} - 1}, 1 - v^{\beta_{j|i} - 1} z_{j|i}(q) + v^{-1}
  z^+_{j|i}(q)\right\}.
\]
\noindent
or
\[
1 - \beta_{j|i} z_{j|i}(q) v^{\beta_{j|i} - 1} < \alpha_{j|i} \leq 1\quad
\mbox{ and } \quad
(1-\beta_{j|i}^{-1})\{\beta_{j|i}z_{j|i}(q)\}^{1/(1-\beta_{j|i})}
(1-\alpha_{j|i})^{-\beta_{j|i}/(1-\beta_{j|i})}+ z^+_{j|i}(q)>0.
\]
\noindent Case II: either
\[
-\alpha_{j|i} \leq \min\left\{ 1,
  1+\beta_{j|i}v^{\beta_{j|i}-1}z_{j|i}(q),
  1+v^{\beta_{j|i}-1}z_{j|i}(q) - v^{-1} z^{-}_{j|i}(q)\right\}
\]
\noindent
or
\[
1 + \beta_{j|i}v^{\beta_{j|i}-1}z_{j|i}(q)<-\alpha_{j|i} \le 1 \quad\mbox{
  and } \quad
(1-\beta_{j|i}^{-1})(-\beta_{j|i}z_{j|i}(q))^{1/(1-\beta_{j|i})}
(1+\alpha_{j|i})^{-\beta_{j|i}/(1-\beta_{j|i})} - z^{-}_{j|i}(q)>0.
\]
where $v>u$ is a value above the maximum observed value of $Y_{i}$ so
that the constraints are imposed only on extrapolations.\ As far as
the selection of $q$ is concerned, \cite{keefpaptawn11} found
empirically that for both cases  conditions were satisfied for all
$q$ if they were each satisfied for both $q=0$ and $q=1$.

% and $q$ such that $z_{j|i}(q)$ corresponds to the largest observed
% value.

%-------------------------------------------------------------------------
\section{Estimation of \cite{hefftawn04} model under stochastic ordering}
\label{sec:stochorder}
%-------------------------------------------------------------------------

%-------------------------------------------------------------------------
\subsection{Quantile Ordering Constraints}
\label{sec:quantorderconstr}
%-------------------------------------------------------------------------

In this paper we exploit the same idea for the construction of the
parameter space of the \cite{hefftawn04} model under the assumption of
stochastic ordering between conditional random variables.\
Specifically, let the $q$th, $q\in[0,1]$, conditional quantile of
$Y_{l}|Y_{j}=x$ and $Y_{k}|Y_{i}=x$, for large $x$, be $y_{l|j}(q)$
and $y_{k|i}(q)$, respectively.\ Under the \cite{hefftawn04} model we
have that $y_{l|j}(q)=\alpha_{l|j}x+x^{\beta_{l|j}}z_{l|j}(q)$ and
$y_{k|i}(q)=\alpha_{k|i}x+x^{\beta_{k|i}}z_{k|i}(q)$.\ Our objective
is to derive constraints under which there is stochastic ordering
between the conditional variables so that the following condition is
always satisfied for all $x$ above a level $v > u$
\begin{equation}
  \label{eq:inequality}
  y_{k|i}(q)\leq y_{l|j}(q),\quad \text{for all } q \in [0,1].
\end{equation}
%---------------------------------------------------------------------
%Theorem for Stochastic Ordering restrictions.
%---------------------------------------------------------------------
The motivation for exploring inequality~(\ref{eq:inequality}) stems
from the dose ordering effect in the joint region of ALT and
TBL. Consider for example, the transformed ALT and TBL with respect to
equation~(\ref{eq:laplace}), and let $y_{2|1}^{A}(q)$ and
$y_{2|1}^{B}(q)$ be the conditional quantiles of TBL given a large
level of ALT for dose $A$ and $B$, respectively. Then under the
assumption of liver toxicity, it is intuitive to consider the natural
ordering of the conditional quantiles $y_{2|1}^{A}(q)\leq
y_{2|1}^{B}(q)$. The following theorem gives conditions under which
two conditional quantiles based on the \cite{hefftawn04} satisfy the
ordering constraint (\ref{eq:inequality}), for a $q \in [0,1]$.
\begin{table}[htbp!]
  \caption{Exclusive conditions for the stationary points of $D(x)$, 
    where $\alpha_{l|j}\geq\alpha_{k|i}$ and
    $s = \left[\left\{\beta_{k|i}(\beta_{k|i}-1)z_{k|i}(q)\right\}/
      \left\{\beta_{l|j}(\beta_{l|j}-1)z_{l|j}(q)\right\}\right]^
    {1/(\beta_{l|j}-\beta_{k|i})}$.}
  \centering
  \begin{tabular}{ c | c | c | c | c }
    \hline
    &number of s.p. & $D'(v)$ & $s$ &$D'(s),~ s\in \mathbb{R}$  \\
    \hline        
    & 0 & $>0$  & complex/real & $(0,\infty)$  \\
    & 1 & $<0$  & complex/real & $(-\infty,\infty)$  \\
    & 2 & $>0$  & real         & $(-\infty,0)$      
    \label{table:sp}
  \end{tabular}
\end{table} 
\begin{theorem}
  \label{theorem:1}
  Let $D(x):\left[v,\infty\right)\rightarrow \mathbb{R}$, $v>0$, such
  that $D(x)=(\alpha_{l|j}-\alpha_{k|i})x +
  x^{\beta_{l|j}}z_{l|j}(q)-x^{\beta_{k|i}}z_{k|i}(q)$, with
  $(\alpha_{l|j},\alpha_{k|i})\in [-1,1]^{2}$,
  $(\beta_{l|j},\beta_{k|i})\in (-\infty,1)^2$, and
  $(z_{l|j},z_{k|i})\in\mathbb{R}^2$. For $v \geq u$ and for all $q
  \in [0,1]$, the ordering constraint~(\ref{eq:inequality}) holds for
  all $x>v$ if $\alpha_{l|j}\geq\alpha_{k|i}$ and for all $q\in [0,1]$,
  either
  \begin{enumerate}
  \item $D(x)$ has no stationary point and $D(v) \geq 0$, or
  \item $D(x)$ has one stationary point $x_{*}>v$ and
    $\min\left\{D(v),D(x_{*})\right\} \geq 0$, or
  \item $\min\left\{D(v),D(x_{*}),D(x_{**})\right\} \geq 0$, where
    $x_{*}$ and $x_{**}$ are the two stationary points of $D(x)$, with
    $\min(x_{*},x_{**})>v$.
  \end{enumerate}

  % \noindent i) the number of the stationary points (s.p.) of $D(x)$,
  % subject to, can be categorised
  % according to the exclusive conditions set out by the columns 2-4 of
  % Table \ref{table:sp},
% \begin{enumerate}
% \item when $D(x)$ has no s.p.,
%   \[
%   \alpha_{k|i} \leq \min\Big\{\alpha_{l|j},\alpha_{l|j}+
%   v^{(\beta_{l|j}-1)}z_{l|j}(q)- v^{(\beta_{k|i}-1)}z_{k|i}(q)\Big\},
%   \]
% \item when $D(x)$ has one s.p.,
%   \begin{IEEEeqnarray*}{rCl}
%     \alpha_{k|i} &\leq& \min\Big\{\alpha_{l|j},\alpha_{l|j}+
%     v^{(\beta_{l|j}-1)}z_{l|j}(q)- v^{(\beta_{k|i}-1)}z_{k|i}(q),\\\\
%     &&\alpha_{l|j}+ x_{*}^{(\beta_{l|j}-1)}z_{l|j}(q)-
%     x_{*}^{(\beta_{k|i}-1)}z_{k|i}(q)\Big\},
%   \end{IEEEeqnarray*}
% \item otherwise,
%   \begin{IEEEeqnarray*}{rCl}
%     \alpha_{k|i} &\leq& \min\Big\{\alpha_{l|j},
%     \alpha_{l|j}+v^{(\beta_{l|j}-1)}z_{l|j}(q)-v^
%     {(\beta_{k|i}-1)}z_{k|i}(q),\\\\
%     &&\alpha_{l|j}+x_{*}^{(\beta_{l|j}-1)}z_{l|j}(q)-x_{*}^
%     {(\beta_{k|i}-1)}z_{k|i}(q),\\\\
%     &&\alpha_{l|j}+x_{**}^{(\beta_{l|j}-1)}z_{l|j}(q)-x_{**}^
%     {(\beta_{k|i}-1)}z_{k|i}(q)\Big\}.
%   \end{IEEEeqnarray*}
% \end{enumerate}
\end{theorem}
%---------------------------------------------------------------------
%Proof of theorem
%---------------------------------------------------------------------
\begin{proof}
  According to Descarte's rule of signs and its extension to
  generalised polynomials \citep{jame06}, $D'(x)=0$ can have at most
  two solutions.\ Therefore $D(x)$ can have at most two stationary
  points.\ Numerical inspection of the function (e.g.\ for
  $\alpha_{l|j}=0.2,\alpha_{k|i}=0.1,\beta_{l|j}=0.2,
  \beta_{k|i}=0.5$, $z_{l|j}(q)=0.6$ and $z_{k|i}(q)=0.6$) shows that
  there can be cases where $D(x)$ has two stationary points.\ The
  cases of Table~\ref{table:sp} follow from noting that $D''(s)=0$ is
  the unique root of $D''(x)$, so that $D'(x)$ has at most one
  stationary point, i.e.\, when $s>v \in \mathbb{R}$ then $s$ is a
  s.p.\ of $D'(x)$, otherwise $s$ is a complex number so that $D'(x)$
  is monotone for $x>v$.  \newline

\noindent The condition $D(x)\geq 0$, for all
$x\in\left[v,\infty\right)$, implies that $\lim_{x \rightarrow
  \infty}D(x)$ can be either 0 or $\infty$. Hence,
  \[
    \alpha_{l|j}\geq\alpha_{k|i}.
  \]
  Categorising the cases with respect to the number of stationary
  points of $D(x)$, we have that $D(x)\geq 0$, for all $x>v$, if and
  only if one of the 3 conditions of Theorem~\ref{theorem:1} (ii)
  holds.\
\end{proof}
From a computational perspective, the constraints follow from the
nature of the $D(x)$ function, i.e.\, one needs to find the stationary
points of $D(x)$ so that estimation of parameters in the
\cite{hefftawn04} model under the quantile ordering assumption can be
carried out.\ The conditions in columns 2-4 of Table~\ref{table:sp}
are necessary and sufficient for the number of stationary points
specified in column 1 and can be checked numerically.\ Also, the
function $D'(x)$ is not linear so closed form roots of $D'(x)=0$ do
not exist. If $D(x)$ has one stationary point then one dimensional
root finding is sufficient to estimate the root of $D'(x)$.\ If $D(x)$
has two stationary points the domain of the function $D(x)$ can be
separated into two subintervals $(v,s)$ and $(s,\infty)$, and in each
interval, one dimensional root finding is sufficient to yield
estimates of these two stationary points.

%-------------------------------------------------------------------------
\subsection{Inference based on stochastic ordering assumptions}
\label{sec:inferencestoch}
%-------------------------------------------------------------------------
Regarding estimation of the \cite{hefftawn04} model under stochastic
ordering, Theorem \ref{theorem:1} provides a set of exclusive cases
where each one shows the number of stationary points that the function
$D(x)$ can have.\ This provides an automatic way for selecting the
associated stochastic ordering condition that is used, jointly with
the constraints~(\ref{eq:keef_constr}) of \cite{keefpaptawn11} for
asymptotic dependence, to constrain the likelihood of the model.\ To
constrain more than two conditional survival curves, e.g., $Y_l\mid
Y_j=x$, $Y_k\mid Y_i=x$$,Y_m\mid Y_p=x $, maximisation of the
likelihood is performed subject to $(\alpha_{l\mid j},\beta_{l\mid
  j},\alpha_{k\mid i},\beta_{k\mid i},\alpha_{m\mid p},\beta_{m\mid
  p})\in R_{lj,ki} \cap R_{ki,mp}$. Here, the set $R_{lj,ki}$, for
example, denotes the parameter space of $(\alpha_{l\mid
  j},\beta_{l\mid j},\alpha_{k\mid i},\beta_{k\mid i})$ subject to
$y_{l|j}^{-}(q) \leq y_{l|j}(q) \leq y_{l|j}^{+}(q)$, $y_{k|i}^{-}(q)
\leq y_{k|i}(q) \leq y_{k|i}^{+}(q)$ and $y_{l \mid j}(q)\geq y_{k
  \mid i}(q)$, for all $x>v$.
\begin{figure}[htbp!]
  \centering
  \includegraphics[scale=1.2,trim= 20 15 20 20]{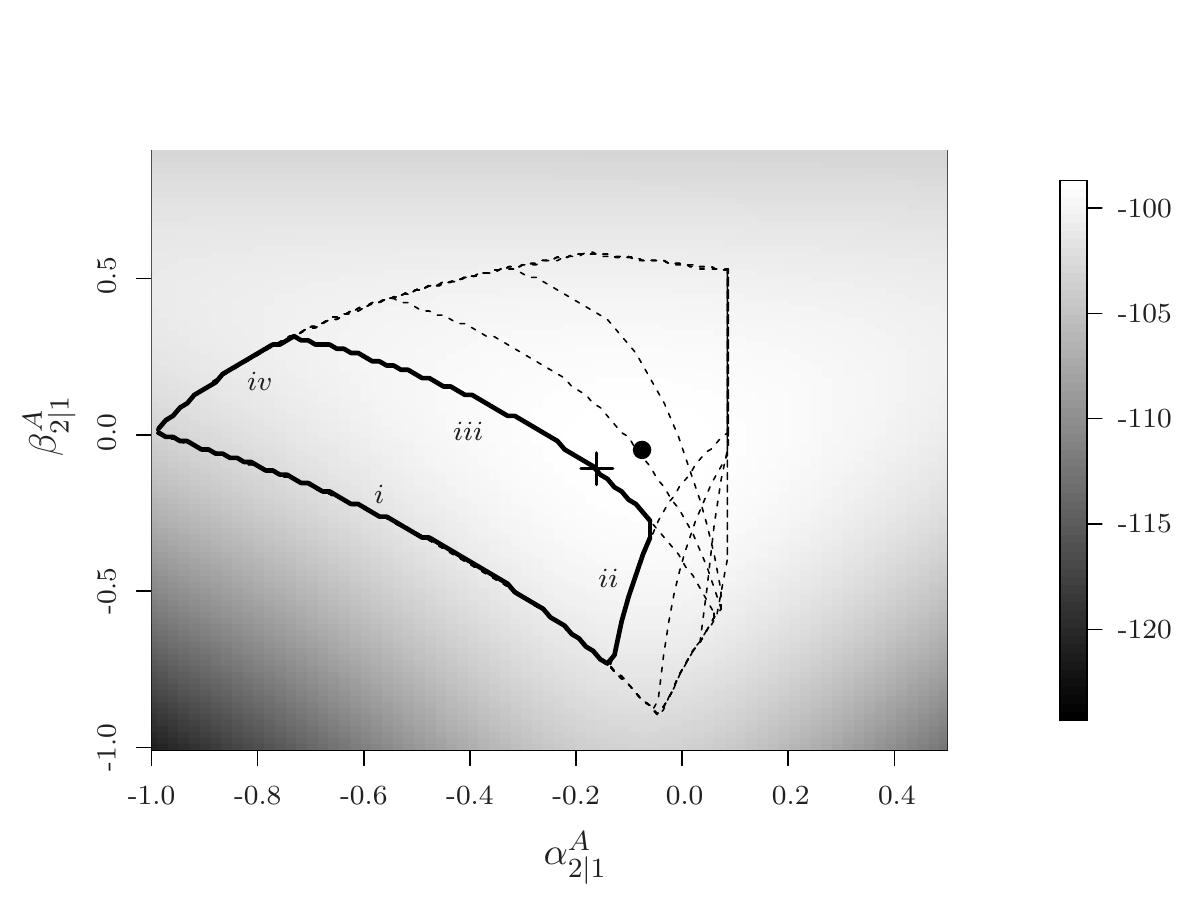}
  % \pstool[scale=0.6,angle=270]{SOplot.eps}{
  %   \psfrag{AA}{\scriptsize $i$}
  %   \psfrag{BB}{\scriptsize $ii$}
  %   \psfrag{CC}{\scriptsize $iii$}
  %   \psfrag{DD}{\scriptsize $iv$}
  %   \psfrag{a}{$\alpha_{2|1}^{A}$} \psfrag{b}{$\beta_{2|1}^{A}$}
  %   \psfrag{-0.2}{\scriptsize -$0.2$} \psfrag{-1.0}{\scriptsize
  %     -$1.0$} \psfrag{-0.8}{\scriptsize -$0.8$}
  %   \psfrag{-0.6}{\scriptsize -$0.6$} \psfrag{-0.4}{\scriptsize
  %     -$0.4$} \psfrag{0.0}{\scriptsize$0.0$}
  %   \psfrag{0.2}{\scriptsize$0.2$}
  %   \psfrag{0.4}{\scriptsize$0.4$}\psfrag{-0.5}{\scriptsize -$0.5$}
  %   \psfrag{0.5}{\scriptsize$0.5$} \psfrag{-100}{\scriptsize -$100$}
  %   \psfrag{-105}{\scriptsize -$105$} \psfrag{-110}{\scriptsize
  %     -$110$}\psfrag{-115}{\scriptsize -$115$} \psfrag{-120}{\scriptsize
  %     -$120$}}
  \caption{Profile log-likelihood surface for dose $A$ parameters
    $(\alpha_{2|1}^{A},\beta_{2|1}^{A})$.\ The solid curves show the
    boundary of the parameter space under the constraints of
    Theorem~\ref{theorem:1} when $q=0$ and $q=1$.\ Dashed curves show
    the constraints of Theorem 1 when $0 < q < 1$, showing these
    constraints are less restrictive than when $q = 0$ and $q = 1$.\
    The dot and cross show estimated parameters for unconstrained and
    constrained estimation respectively. The labels $i$, $ii$, $iii$
    and $iv$ are explained in Section~\ref{sec:inferencestoch}.}
  \label{fig:SO}
\end{figure}
For the required stochastic ordering (\ref{eq:inequality}) constraint
we numerically found that the conditions of Theorem~\ref{theorem:1}
were satisfied for all $q$ if they were satisfied for both $q = 0$ and
$1$.\ To illustrate this feature, Figure \ref{fig:SO} shows the
profile log-likelihood surface of the conditional dependence model
parameters of TBL given ALT for dose $A$, denoted by
$\alpha_{2|1}^{A}$ and $\beta_{2|1}^{A}$, under the assumption that
the conditional quantile of dose $A$ is smaller than the conditional
quantile of dose $B$.\ The solid lines correspond to the joint $q=0$
and $q=1$ constraints whereas the dashed lines correspond to $0<q<1$
constraints.\ In particular, the shape of the constraints is
quasi-trapezoid with sides highlighted on the figure for the joint
$q=0$ and $q=1$ case by $i,ii,iii$ and $iv$.\ Sides $i$ and $ii$ are
affected by quantiles near 0, with dashed lines showing, in the bottom
right area of the figure and from left to right, the constraints
induced by the $0,0.02,0.05$ and $0.1$ quantiles.\ All side $ii$
constraints join at a point and extend to become the upper
boundary. Sides $iii$ and $iv$ are mostly affected by quantiles near
$1$, with dashed lines showing, from bottom to top, the constraints
induced by the $1, 0.9$ and $0.8$ quantiles.\ All side $iii$
constraints join at a point and extend to become the lower
boundary. The constraints induced by small quantiles affect different
areas of the parameter space than those induced by larger quantiles
and, as illustrated by Figure \ref{fig:SO}, crossover is possible.\
However, the parameter space obtained from the joint $q=0$ and $q=1$
constraints is nested in the parameter spaces obtained from the
$0<q<1$ constraints.

%-------------------------------------------------------------------------
\subsection{Tests of ordering hypotheses in the conditional tail}
\label{sec:hypothesis}
%-------------------------------------------------------------------------
The ordering constraints developed in
Section~\ref{sec:quantorderconstr} permit the testing of hypothesis of
ordering of two or more conditional survival curves. Suppose that the
survival curves of $Y_{l}|Y_{j}=x$ and $Y_{k}|Y_{i}=x$ are ordered for
all $x$ greater than a large threshold $v$, and let
$\boldsymbol{\theta} =
(\alpha_{l|j},\beta_{l|j},\alpha_{k|i},\beta_{k|i})$. We focus on
testing the composite hypothesis of ordering between the two
conditional tails, i.e.\,
\[
H_0: \boldsymbol{\theta} \in R_{lj,ki} \qquad \text{vs } \qquad
H_1:\boldsymbol{\theta} \in R_{lj,ki}^c,
\]
where $R_{lj,ki}^c$ is the complement of the set $R_{lj,ki}$, i.e.\,
$R_{lj,ki}$ is the constrained space defined in
Section~\ref{sec:inferencestoch}.

Let $\hat{\boldsymbol{\theta}}_0$ and $\hat{\boldsymbol{\theta}}$
be the maximum likelihood estimators subject to
$\boldsymbol{\theta}\in R_{lj,ki}$ and $\boldsymbol{\theta}\in
R_{lj,ki} \cup R_{lj,ki}^c$, respectively. To obtain a test for the
composite hypothesis of ordering in the tail $H_0:
\boldsymbol{\theta}\in R_{lj,ki}$, in the presence of additional
nuissance parameters $(\mu_{l\mid j},\mu_{k \mid i}, \sigma_{l\mid
  j},\sigma_{k \mid i})$, we compare how much larger the profile
log-likelihood is achieved at the maximum $\hat{\boldsymbol{\theta}}$,
than at the null hypothesis $\hat{\boldsymbol{\theta}}_0$.\ Thus, the
generalised likelihood ratio test criterion is used and its null
hypothesis distribution is obtained by simulation, i.e.\, by employing
the first three stages of Algorithm~\ref{alg:1}. Simulation of the
distribution of the likelihood ratio test statistic is necessary for
this situation since the approximate distribution cannot be obtained
by analytic derivation \citep{Cox06}.

Hypotheses and tests of ordering for more than two conditional random
variables are obtained analogously with the notation outlined in
Section~\ref{sec:inferencestoch}.
% [\color{red}Jon, I am not quite sure if you will like Sections 3.2 and
% 3.3. Notation seems complicated and I am not sure if it reads
% well. Any comment would be very helpful.\color{black}]

%-------------------------------------------------------------------------
\section{Simulation study of ordering constraints}
\label{sec:simulation}
%-------------------------------------------------------------------------
\subsection{Design}
\label{sec:design}
The impact of the proposed constraints of Section
\ref{sec:quantorderconstr} and the \cite{keefpaptawn11} constraints is
illustrated with a simulation study.\ We examine the performance of
conditional quantile estimates using simulated datasets from the
limiting representation of three bivariate copula models with Laplace
marginals, namely the logistic, inverted logistic and standard
Gaussian copulas, with dependence parameters $\lambda \in (0,1]$,
$\kappa\in (0,1]$ and $\rho>0$, respectively. The first two models are
also known as the bivariate Gumbel and its corresponding survival
copula. The models can be found in \cite{hefftawn04}, Section
8.\ % Denote by $(Y_1,Y_2)$ a random
% vector arising from one of these copula models.\
We simulate pairs of observations conditionally on $Y_1$ exceeding a
finite threshold $u$ from the \emph{exact} form of the limiting
conditional dependence model.\ Explicitly, we assume that the
conditional distribution function of $Y_2|Y_1>u$, for finite $u$, is
equal to the actual limiting distribution function that is implied by
expression (\ref{eq:G}) for each model, i.e.\, we assume that
\begin{equation}
Y_2 = \alpha_{2|1}Y_1 + Y_1^{\beta_{2|1}}Z_{2|1}, \quad Z_{2|1}\sim
G_{2|1}, \quad Y_1>u \quad \text{and } \quad Z_{2|1}\text{ independent
  of } Y_1
\label{eq:exact_model}
\end{equation}
with $\alpha_{2|1}$, $\beta_{2|1}$ and $G_{2|1}$ chosen such that
expression~({\ref{eq:G}}) holds.\ Algorithm~\ref{algo2} describes the
simulation procedure used in this study for the bivariate case.\ The
normalising parameters and the residual distribution of the limiting
representation~(\ref{eq:G}) are summarised in
Table~(\ref{tab:exact_models}).\
\begin{algorithm}
\caption{Simulation}
\label{algo2}
\begin{algorithmic}[1] 
  \State Set $I=1$, $\left(\alpha_{2|1},\beta_{2|1}\right) \in
  (-1,1)\times(-\infty,1)$ and $N\in \mathbb{N}$.

  \vspace{5pt} \State Simulate $y_{1,I}$ from a Laplace distribution
  conditional on exceeding a threshold $u$.

  \vspace{5pt}
  
  \State Simulate $z_{2|1,I}$ independently of $y_{1,I}$ from the true
  limiting distribution $G_{2|1}$

  \State Set $y_{2,I}=\alpha_{2|1}y_{1,I}+
  y_{1,I}^{\beta_{2|1}}z_{2|1,I}$.
  
  \vspace{5pt}
  
  \State If $I<N$ set $I=I+1$ and go to step 2; otherwise return
  $\left(\boldsymbol{y}_{1},\boldsymbol{y}_{2}\right)$, where
  $\boldsymbol{y}_{1}=
  \left(y_{1,1},...,y_{1,N}\right)',\boldsymbol{y}_{2}=\left(y_{2,1},...,y_{2,N}\right)'$.
\end{algorithmic}
\end{algorithm}
% \subsection{Theoretical models}
% \label{sec:thmodels}
% The three copula models we consider in the simulation study are the
% bivariate extreme value logistic, the inverted logistic and the
% Gaussian with distribution functions given respectively by
% \begin{IEEEeqnarray*}{rCl}
%   \mbbP\left(Y_1 \leq x, Y_2 \leq y\right)&=&
%   \exp\left(-\left[\left\{-\log F(x) \right\}^{1/\kappa} +
%       \left\{-\log F(y) \right\}^{1/\kappa}\right]^{\kappa}\right),\quad \kappa \in (0,1],\\\\
%   \mbbP\left(Y_1 \leq x, Y_2 \leq y\right)&=&\mbbP\left(T_1>-x,T_2>-y\right),\quad \text{where } (T_1,T_2)\sim\text{logistic}(\kappa),\\\\
%   \mbbP\left(Y_1 \leq x, Y_2 \leq y\right)&=& \Phi_{\Sigma}\left[
%     \Phi^{-1}\left\{F(x)\right\},\Phi^{-1}\left\{F(y)\right\}\right],
% \end{IEEEeqnarray*}
% where $F$, $\Phi$ and $\Phi_{\Sigma}$ denote the d.f. of the standard
% Laplace, the standard normal and the standard bivariate Gaussian with
% correlation matrix $\Sigma=(\sigma_{i,j})$ with $\sigma_{i,i}=1$ and
% $\sigma_{i,j}=\rho$, $\rho\in(-1,1)$, for $i\neq j$. The normalising
% functions and the residual distribution $G_{2|1}$ of the limiting
% representation~(\ref{eq:G}) are summarised in
% Table~(\ref{tab:exact_models})

\begin{table}[!htbp]
  \caption{\cite{hefftawn04} normalising constants $\alpha_{2|1}$, $\beta_{2|1}$, and limiting distribution $G_{2|1}$ for the bivariate extreme value logistic, inverted logistic and standard bivariate Gaussian copula with dependence parameters $\lambda \in (0,1]$, $\kappa\in(0,1]$ and $\rho>0$, respectively.}
  \label{tab:exact_models}
  \begin{center}
    \begin{tabular}{c|c|c|c}
      \hline
      Model & $\alpha_{2|1}$ & $\beta_{2|1}$ & $G_{2|1}(z)$\\
      \hline
      logistic             & 1 & 0  & $\left\{1+\exp\left(-z/\lambda\right)
      \right\}^{\lambda-1}$ \\
      inverted logistic    & 0  & $1-\kappa$ & $1-\exp\left(-\kappa z^{1/\kappa}\right)$ \\
      Gaussian             & $\rho^2$  & $1/2$ & N$\{0,2\rho^2(1-\rho^2)\}$ \\
      \hline
    \end{tabular}
  \end{center}
\end{table}
The values of the parameters $\lambda$, $\kappa$ and $\rho$ used in
the simulation study are chosen such that the simulated data preserve
the stochastic ordering feature. For example, in the logistic copula
case, dependence increases as the value of $\lambda$ decreases which
implies larger joint survival probabilities as $\lambda$ decreases. We
thus simulate pairs of observations from the exact conditional
dependence model~(\ref{eq:exact_model}) under $\lambda=0.6$ and
$\lambda=0.9$. For the asymptotically independent models, two pairs of
parameter values are used in the simulation study, i.e.\, for the
inverted logistic copula we use $\kappa=(0.3,0.7)$ and
$\kappa=(1,0.415)$, and for the Gaussian copula we use
$\rho=(0.3,0.7)$ and $\rho=(0,0.5)$.\ The second pair of $\kappa$ and
$\rho$ values is used for comparisons between the two asymptotically
independent models, since the coefficient of tail dependence of
\cite{ledtawn96}, is the same for the inverted logistic and bivariate
Gaussian copula models with $(\kappa,\rho)=(0.415,0.5)$ and $(0,1)$,
respectively. The coefficient of tail dependence is a key summary
measure of extremal tail dependence between the variables $Y_{1}$ and
$Y_{2}$.
  
The conditional quantile estimates are obtained from the original
\cite{hefftawn04} model, the constrained model of \cite{keefpaptawn11}
and the constrained model described in Section
\ref{sec:quantorderconstr}.\ We refer to these models as
Heffernan--Tawn (HT), asymptotic dependence (AD) and stochastic
ordering (SO), respectively.\ The stochastic ordering constraints are
imposed on pairs of observations with different dependence
parameters. The performance of the estimates is assessed with the
Monte Carlo estimate of the root mean square error.\ To be specific
let $y(q)$ and $\hat{y}(q)$ be the true conditional quantile and its
model-based estimate.\ The Monte Carlo estimate of the root mean
square error is
\[
\sqrt{\frac{1}{m}\sum_{i=1}^{m}\left\{\hat{y}_{i}(q) -
    y(q)\right\}^{2}},
\]
where $\hat{y}_{1}(q),...,\hat{y}_{m}(q)$ denotes a Monte Carlo sample
of the conditional quantile estimates and each estimate is obtained
from a simulated sample of $N \in \mathbb{N}$ pairs of observations,
$\left\{\left(\boldsymbol{y}_{1},\boldsymbol{y}_{2}\right):y_{1,i}>u,
  i=1,...,N\right\}$.\ The values $N=500$ and $m=1000$ are used. The
root mean square error is estimated for $q=0.2$ and 0.8, and the
conditioning levels $x_{0.95}$ and $x_{0.999}$, where $x_p$ denotes
the $p$-th quantile of the standard Laplace distribution.\ Comparisons
are made on the basis of ratios of RMSEs between estimates from
different models.\ For the constrained models we tabulate the
percentage of the Monte Carlo samples where estimates changed with
respect to the original \cite{hefftawn04} model.
\subsection{Results of simulation}
\label{sec:simresults}
Table \ref{tab:perc} shows the percentage of the estimates that
changed with respect to one of the three reference models (HT, AD,
SO).\ The imposition of constraints to the parameter space of the HT
model alters estimates particularly in the asymptotically independent
models and less in the asymptotically dependent model.\ In particular,
the larger changes occur when variables are highly dependent except
for the logistic model. Regarding the logistic model, the percentage
of changes in the first two rows appear similar for different cases of
the parameter values compared to the other models.\ This feature stems
from the model specification which specifies the same norming
parameters for both cases of $\lambda$.\ Additionally, small changes
occur within the asymptotically independent models especially when the
variables do not possess strong dependence.\ For the second pair of
parameter values for the inverted logistic and Gaussian copulas, the
changes in parameter estimates do not occur at a similar rate when
dependence between variables is present ($\kappa=0.415$, $\rho=0.5$).\
We therefore conclude that the constraints induced by the AD and SO
models are not only related to the level of dependence but to the
dependence structure as well.
\begin{table}[!htbp]
  \caption{Percentage of estimates (\%) that changed with respect 
    to a reference model.\ First row: percentage of AD
    estimates different from the HT estimates, second row: 
    percentage of SO estimates different from the HT estimates,
    third row: percentage of SO estimates different from AD estimates. 
    Columns show the corresponding model from Table~\ref{tab:exact_models}
    used in the simulation.}
  \label{tab:perc}
  \begin{center}
    \begin{tabular}{c|cc|cc|cc|cc|cc}
      \hline
      & \multicolumn{2}{c|}{logistic} & \multicolumn{4}{c|}{inverted logistic} & 
      \multicolumn{4}{c}{Gaussian}  \\       
      \hline
      $\lambda$, $\kappa$ and $\rho$& $ 0.6$ & $0.9$ & $0.3$ & $0.7$ &
      $0.415$ &$1$ &
      $0.7$ & $0.3$ & $0.5$ & $0$ \\ 
      \hline
      AD-HT & 29 & 27 & 63 & 10 & 41 & 0.3 & 36  & 0 & 6 & 1 \\ 
      \hline
      SO-HT & 54 & 56 & 77 & 45 & 68 & 47 & 42   & 10 &30  & 25\\ 
      \hline
      SO-AD & 33 & 46 & 43 & 44 & 47 & 47 &9     & 9  & 24 & 25\\ 
    \end{tabular}
  \end{center}
\end{table}
Table~\ref{tab:ratio_RMSE} shows the ratio of the Monte Carlo root
mean square error, of the conditional quantile estimates obtained from
the three copula models.\ An increase in efficiency under the
imposition of the constrained models AD and SO is observed for nearly
all conditional quantile estimates in the asymptotically independent
models. The highest reduction in RMSE is achieved by the SO model in
the inverted logistic copula, a feature which is also consistent with
the higher percentage of change in estimates as shown in
Table~\ref{tab:perc}. The conclusion for the asymptotically
independent models is that the efficiency of the conditional quantile
estimates is, in decreasing order, SO, AD and HT. Regarding the
asymptotically dependent logistic copula, constrained models appear to
be less efficient than the HT model and the efficiency of the
conditional quantile estimates is, in decreasing order, HT, AD and SO.
\begin{table}[!htbp]
  \caption{Ratio of the Monte Carlo root
    mean square error of the conditional quantile estimates $\hat{y}_{2|1}(q)$
    obtained from the HT, AD and SO models. Results are reported for $q=0.2,0.5$ 
    and the conditioning levels $x_{0.95}$ and $x_{0.999}$. 
    The value $x_p$ here denotes the $p$-th quantile of the standard Laplace distribution.}
  \label{tab:ratio_RMSE}
  \begin{center}
    \scalebox{0.8} {
    \begin{tabular}{c|ccc|cccc|cccc}
      \hline\hline
      & & \multicolumn{2}{c|}{logistic copula} & \multicolumn{4}{c|}{inverted logistic copula} & 
      \multicolumn{4}{c}{Gaussian copula}  \\
      \hline\hline
      $\lambda$, $\kappa$ and $\rho$& & $ 0.6$ & $0.9$ & $0.3$ & $0.7$ &
      $0.415$ &$1$ &
      $0.7$ & $0.3$ & $0.5$ & $0$ \\ 
      \hline
      \multicolumn{12}{c}{$q=0.2$} \\
      \hline
      \multirow{3}{*} {AD/HT} & $x_{0.95}$&1.10& 1.02 & 0.99 & 1.00 &0.98 & 1.00& 1.00 & 1.00 & 0.99&0.99 \\
      % & $x_{0.99}$ & 1.00 & 1.08 & 0.99&0.97 & 0.98 & 0.98 & 1.00 & 1.00 &1.00 &1.00\\
      & $x_{0.999}$ & 1.10 & 1.04&1.00 &0.96 & 0.97 & 0.98 & 0.99 &1.00 & 0.99&0.99\\
      \hline
      \multirow{3}{*} {SO/HT}& $x_{0.95}$ & 1.20 & 0.97& 0.93&  0.99& 0.97 & 0.96 &0.98 & 0.99 & 0.99&1.00\\
      % & $x_{0.99}$ & 0.99  & 1.24 &0.95 &0.97 & 0.94 & 0.96 & 0.97 & 0.99 &0.99 &1.00\\
      & $x_{0.999}$ & 0.96 & 1.12& 0.96 & 0.94& 0.94 & 0.98 & 0.97 & 0.99 & 0.97&0.97\\
      \hline
      \multirow{3}{*} {SO/AD}& $x_{0.95}$ & 1.10& 0.94& 0.93 & 0.99 & 0.99 & 0.96 & 0.98 & 0.99 & 0.99&1.00 \\
      % & $x_{0.99}$ & 0.96&  1.14 & 0.95 & 0.99& 0.95& 0.97 & 0.97 & 0.99 & 0.99 &1.00\\
      & $x_{0.999}$ & 1.01& 1.07& 0.96& 0.98 & 0.96 & 0.99 & 0.98 & 0.99 & 0.97 &0.97\\
      \hline
      % \multicolumn{12}{c}{$q=0.5$} \\
      % \hline
      % \multirow{3}{*} {AD/HT}& $x_{0.95}$ & 0.99 &1.05&1.02 & 1.02 & 0.98 & 1.00& 1.00 & 1.00 &0.99 & 1.00\\
      % % & $x_{0.99}$ &1.09&1.08&1.01 & 0.96 & 0.98 & 0.98 &1.01 & 1.00 & 1.00&1.00\\
      % & $x_{0.999}$ &1.11&1.08&1.00 & 0.89& 0.94 & 0.95 &1.00 & 1.00 & 0.99&0.97\\
      % \hline
      % \multirow{3}{*} {SO/HT}& $x_{0.95}$ &0.78&1.43&0.85 & 1.02 & 0.89   & 0.98 & 0.94 & 0.99 & 0.94&1.01\\
      % % & $x_{0.99}$ &1.09&1.29&0.79 & 0.96  & 0.82 & 0.93 & 0.90 & 0.98&0.93&1.00\\
      % & $x_{0.999}$ &1.15&1.27&0.84 & 0.85 & 0.82 & 0.92 & 0.86 & 0.98&0.90&0.90\\
      % \hline
      % \multirow{3}{*} {SO/AD}& $x_{0.95}$ &0.78&1.35&0.84 & 0.99 & 0.89& 0.97 & 0.94 & 0.99 & 0.94&1.01\\
      % % & $x_{0.99}$ &0.99&1.19&0.78 & 0.99 & 0.83 & 0.94 & 0.89& 0.98 & 0.93 &1.00\\
      % & $x_{0.999}$ &1.03&1.17&0.84 & 0.95& 0.86 & 0.96 & 0.86& 0.98 & 0.91 &0.92\\
      % \hline
      \multicolumn{12}{c}{$q=0.8$} \\
      \hline
      \multirow{3}{*} {AD/HT}& $x_{0.95}$ &1.04&1.05&1.02 & 1.03 & 1.00 & 0.99 & 1.00 & 1.00 & 0.99 &1.00\\
      % & $x_{0.99}$ &1.10&1.08&1.03 &0.98 & 0.99 & 0.99 & 1.01 & 1.00 & 1.00 &0.99\\
      & $x_{0.999}$ &1.14&1.09&0.81 &0.86 & 0.94& 0.93 & 1.00& 1.00& 0.99&0.95\\
      \hline
      \multirow{3}{*} {SO/HT}& $x_{0.95}$ &0.97&1.25&0.84 & 1.04 & 0.84 & 0.99 & 0.91  & 0.99 & 0.92 &1.02\\
      % & $x_{0.99}$ &1.12&1.29&0.73 & 0.99 & 0.77 & 0.95 & 0.87 & 0.98 & 0.90&0.99\\
      & $x_{0.999}$ &1.18&1.31&0.71 & 0.82 & 0.75& 0.88 & 0.81& 0.97& 0.88&0.88\\
      \hline
      \multirow{3}{*}  {SO/AD}& $x_{0.95}$ &0.93&1.19&0.81 & 1.00 & 0.84 & 1.00 & 0.91 & 0.99 & 0.92 &1.01\\
      % & $x_{0.99}$ &1.01&1.19&0.71 & 1.00 & 0.78 & 0.95& 0.86 &
      % 0.98&0.90 &1.00\\
      & $x_{0.999}$ &1.04&1.19&0.74 & 0.94 & 0.79 & 0.94& 0.80 & 0.97& 0.88&0.92\\
      \hline
    \end{tabular}
  }
  \end{center}
\end{table}
\newpage
%-------------------------------------------------------------------------
\section{Application: drug-induced liver injury}
\label{sec:application}
%-------------------------------------------------------------------------

%---------------------------------------------------------------------
\subsection{Preprocessing and outline of analysis}
\label{sec:Introapp}
%---------------------------------------------------------------------
The data that we consider in this study relates to a sample of 606
patients that were issued a drug linked to liver injury in a
randomised, parallel group, double blind Phase 3 clinical study.\ ALT
and TBL measurements were collected from all patients at baseline
(prior to treatment) and post-baseline (after 6 weeks of treatment)
periods.\ Let $V_{i,B}^{j}$ and $V_{i,P}^{j}$ be the $i$-th baseline
and post-baseline laboratory variable respectively, measured at dose
$j=A,B,C$ and $D$.\ We use $i=1,2$ to denote the ALT and TBL,
respectively. 

Instead of working with the raw data, the \cite{boxcox64}
transformation is applied initially to stabilise the heterogeneity
observed in the samples.\ For this dataset we apply the
log-transformation and we denote the transformed data by
$W_{i,B}^{j}=\log(V_{i,B}^{j})$ and $W_{i,P}^{j}=\log(V_{i,P}^{j})$.\
Consequently, we use a robust linear regression model of the
log-post-baseline on the log-baseline variable to adjust for the
baseline effect, i.e.\, in its simplest form, the robust linear
regression of $W_{i,P}^{j}$ on $W_{i,B}^{j}$ is
\begin{equation}
W_{i,P}^{j}= \gamma_{i}^{j} + \delta_{i}^{j} W_{i,B}^{j} +
X_{i}^{j},\quad i=1,2, \quad j=A,\hdots,D,
\label{eq:qreg}
\end{equation}
where $(\gamma_{i}^{j},\delta_{i}^{j}) \in \mbbR^2$ and $X_{i}^{j}$ is
a zero mean error random variable.\ Here we use median quantile
regression \citep{regquant} which is equivalent to assuming that the
error random variable $X_{i}^{j}$ follows the Laplace distribution
with zero location constant scale parameters \citep{yumoy01}.\ The
parameter estimates $(\hat{\gamma}_{1}^{j})$,
$(\hat{\gamma}_{2}^{j})$, $(\hat{\delta}_{1}^{j})$ and
$(\hat{\delta}_{2}^{j})$, $j=A,B,C,D$, were found to be all
significantly different from 0 and equal to $(0.48,0.53,0.46,0.99),
(0.40,0.58,0.40,0.58)$, $(0.86,0.84,0.91,0.74)$ and
$(0.84,0.77,0.86,0.81)$, all indicating positive association of
post-baseline with baseline.

Our approach is based on the basic model structure of
\cite{southheff12b}, i.e.\, the extremal dependence of
$(X_{1}^j,X_{2}^j)$ is estimated from the \cite{hefftawn04}
conditional dependence model whereas the log-baseline variables
$W_{1,B}^{j}$ and $W_{2,B}^{j}$ are modelled independently for each
dose.\ Under the assumption of independence between $X_{i}^{j}$ and
$W_{i,B}^{j}$ simulated samples of the post-baseline variables can be
generated.\ In this example, the maximum Spearman's correlation
observed was 0.10 and corresponds to the pair $W_{2,B}^B$ and $X_2^B$,
whereas all other combinations gave values lower than 0.07. The exact
procedure of the simulation is straightforward, i.e.\, residual and
baseline samples are generated from their models and are combined in
equation~(\ref{eq:qreg}), with $\gamma_{i}^{j}$ and $\delta_{i}^{j}$
replaced by their corresponding maximum likelihood estimates, to
produce simulated samples for the log-post-baseline variable
$W_{i,P}^{j}$.\ The simulated sample is then back-transformed to its
original scale using the inverse Box-Cox transformation.

The key differences between our modelling procedure and
\cite{southheff12b} are related to the modelling of the baseline and
the estimation of the conditional dependence model parameters.\
Firstly, for each baseline variable we implement the univariate
semi-parametric model of \cite{coletawn94} as described in
Section~\ref{sec:Laptrans} by equation~(\ref{eq:coletawn}) whereas
\cite{southheff12b} use the empirical distribution function.\ Our
motivation for modelling the tail of the baseline variable stems from
the fact that it is likely to observe higher baseline ALT and TBL in
the population (post-marketing period) than in the clinical trial
(pre-marketing period).\ Therefore, tail modelling of the baseline is
key to the simulation process as it incorporates a natural source of
extremity through model-based extrapolation.\ Results from the
univariate analysis are not presented in this paper but similar
analyses can be found in \cite{southheff12} and \cite{paptawn12}. 

In Section~\ref{sec:appHT} we test and subsequently select the
stochastic ordering model developed in Section
\ref{sec:stochorder}. The effect of the ordering constraints is
illustrated via estimates of conditional quantiles for all doses and
results are compared with the unconstrained estimates obtained from
the HT model.\ We proceed to the prediction of the probability of
extreme joint elevations by simulating post-baseline laboratory data
of hypothetical populations of size $200000$ using the fitted marginal
and conditional dependence models. The assessment of the uncertainty
of the estimates of extreme quantities of interest is performed via
the bootstrap procedure.

\subsection{Hypothesis testing and selection of dependence model}
\label{sec:appHT}
Let $Y_{1}^{j}$ and $Y_{2}^{j}$ be the transformed, with respect to
equation (\ref{eq:laplace}), residuals $X_{1}^j$ and $X_2^j$ for each
dose $j$.\ Figure~\ref{fig:data} shows the bivariate scatterplots of
$Y_{2}^{j}$ against $Y_{1}^{j}$ for all dose levels. The tail
dependence between the residual ALT and TBL variables appears to be
very weak for all dose levels and a direct conclusion regarding the
stochastic ordering effect cannot be made on the basis of
Figure~\ref{fig:data}.\ This is also justified by the estimated
${\chi}$ and $\bar{\chi}$ measures of tail dependence
\citep{coleheff99} which are 0 for all doses.  

To assess the ordering assumption, we use the likelihood ratio
criterion described in Section~\ref{sec:hypothesis}, and test at the
significance level of 5\%, the hypotheses of ordered dose dependence
in the conditional distributions of ALT and TBL given that TBL and ALT
exceed a large threshold $v$, respectively. For the SO model we
selected a range of values $v$ above 5, the $99.7\%$ quantile of the
Laplace distribution.\ Similar results where obtained from all
thresholds and here we report the output for $v=5$.
Figure~\ref{fig:lrt} shows the simulated distribution of the
likelihood ratio test statistic under the null hypothesis of ordered
dependence. Both histograms imply that we cannot reject the null
hypothesis at 5\% with stronger evidence for the distribution of TBL
given large ALT.\ The p-values are approximately 0.43 and 0.15,
respectively. The effect of constraining the parameter space to impose
the stochastic ordering assumption between all dose levels is shown in
Figure \ref{fig:condquantiles} via the conditional quantile estimates
obtained from the SO model.\ A weak lack of ordering appears from the
estimated conditional quantiles of TBL given ALT from the HT model as
shown in Figure~\ref{fig:condquantiles} in the standard Laplace
scale.\ The estimates of the median conditional quantiles from the HT
model are ordered above approximately the conditioning level $4$
whereas the minimum and maximum conditional quantile estimates exhibit
a lack of ordering for the majority of the conditioning levels. The
imposition of the ordering constraints induces changes in all
conditional quantile estimates which satisfy the ordering assumption
above the conditioning level $5$. The most important change in the
quantile estimates is observed for dose $A$ which are considerably
smaller than the HT estimates, when $q=1$.

The focus is placed now on the prediction of joint elevations of ALT
and TBL. As stated by \cite{FDAliver} and mentioned earlier in
Section~\ref{sec:intro}, DILI is associated with ALT and TBL exceeding
the 3$\times$ULN$_{\text{ALT}}$ and 2$\times$ULN$_{\text{TBL}}$
respectively. For ALT, the ULN is taken to be 36 units/litre and for
TBL is 21 $\mu$mol/litre. Let $p^{j}(x,y)$ be the joint survival
probability of $\{\text{ALT}>x \cap \text{TBL}>y\}$ at dose level
$j=A,B,C$ or $D$, i.e.\,
\begin{equation}
p^{j}(x,y):=\mbbP(V_{1,P}^{j} > x, V_{2,P}^{j} > y), \quad x,y\in\mbbR.
\label{eq:survprobs}
\end{equation}
To estimate the survival probability~(\ref{eq:survprobs}) we follow
the approach of \cite{southheff12b}, also mentioned earlier in
Section~\ref{sec:Introapp} and simulate $N=200000$ post-baseline
samples. For each dose level, $N$ baseline samples are generated from
the semi-parametric model~(\ref{eq:coletawn}) and are subsequently
combined with $N$ generated residual samples from the SO constrained
\cite{hefftawn04} model in equation~(\ref{eq:qreg}), with
$\gamma_{i}^{j}$ and $\delta_{i}^{j}$ replaced by their corresponding
maximum likelihood estimates, to produce simulated samples for the
log-post-baseline variable $W_{i,P}^{j}$.\ The simulated sample is
then back-transformed to its original scale and the survival
probability~(\ref{eq:survprobs}) is estimated empirically. To assess
the uncertainty of the estimates, this procedure is repeated $1000$
times and $95\%$ equal-tail confidence intervals are obtained from the
bootstrap distribution of each estimate.

\begin{figure}[htpb!]
\centering 
\includegraphics[scale=1,trim= 20 15 20 20]{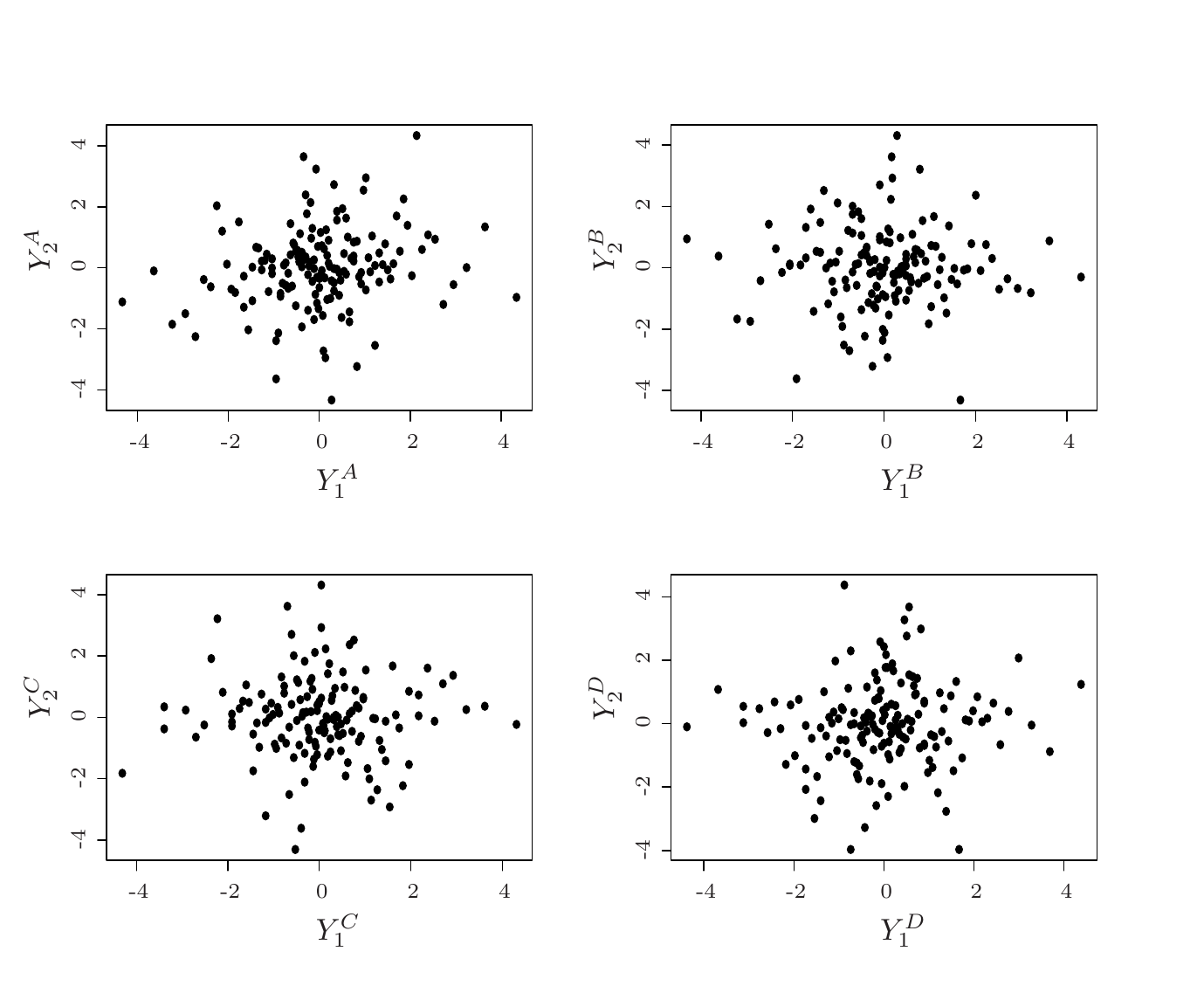}
% \pstool[width=4.5in,height=5.5in,angle=270]{resdata.eps}{
%     \psfrag{a}{$Y_{1}^{A}$} \psfrag{b}{$Y_{2}^{A}$}
%     \psfrag{c}{$Y_{1}^{B}$} \psfrag{d}{$Y_{2}^{B}$}
%     \psfrag{e}{$Y_{1}^{C}$} \psfrag{f}{$Y_{2}^{C}$}
%     \psfrag{g}{$Y_{1}^{D}$} \psfrag{h}{$Y_{2}^{D}$}
%     \psfrag{-4}{\scriptsize -$4$} \psfrag{-2}{\scriptsize -$2$}
%     \psfrag{0}{\scriptsize$0$} \psfrag{2}{\scriptsize$2$}
%     \psfrag{4}{\scriptsize$4$} }  
\caption{Scatterplots of $Y_{2}^{j}$ against $Y_{1}^{j}$, for
  $j=A,...,D$}
\label{fig:data}
\end{figure}
\newpage
Figure~\ref{fig:predictions} shows the estimated survival
probabilities $p^{j}(x,y)$ for $x=3\times\text{ULN}_{\text{ALT}}$ and
variable $y$. For comparisons, estimates are reported from the SO and
HT models. The imposition of the constraints induces changes in all
estimates. In particular, the survival probability estimates from the
SO model are lower than the HT model for all doses, especially in the
region $\{y:y<30\}$. This behaviour also implies changes in the upper
tail and in the joint region of DILI, i.e.\, when $x=3\times
\text{ULN}_{\text{ALT}}$ and $x=2\times \text{ULN}_{\text{TBL}}$.

\begin{figure}[htpb!]
\centering 
\includegraphics[scale=.7,trim= 20 15 20 20]{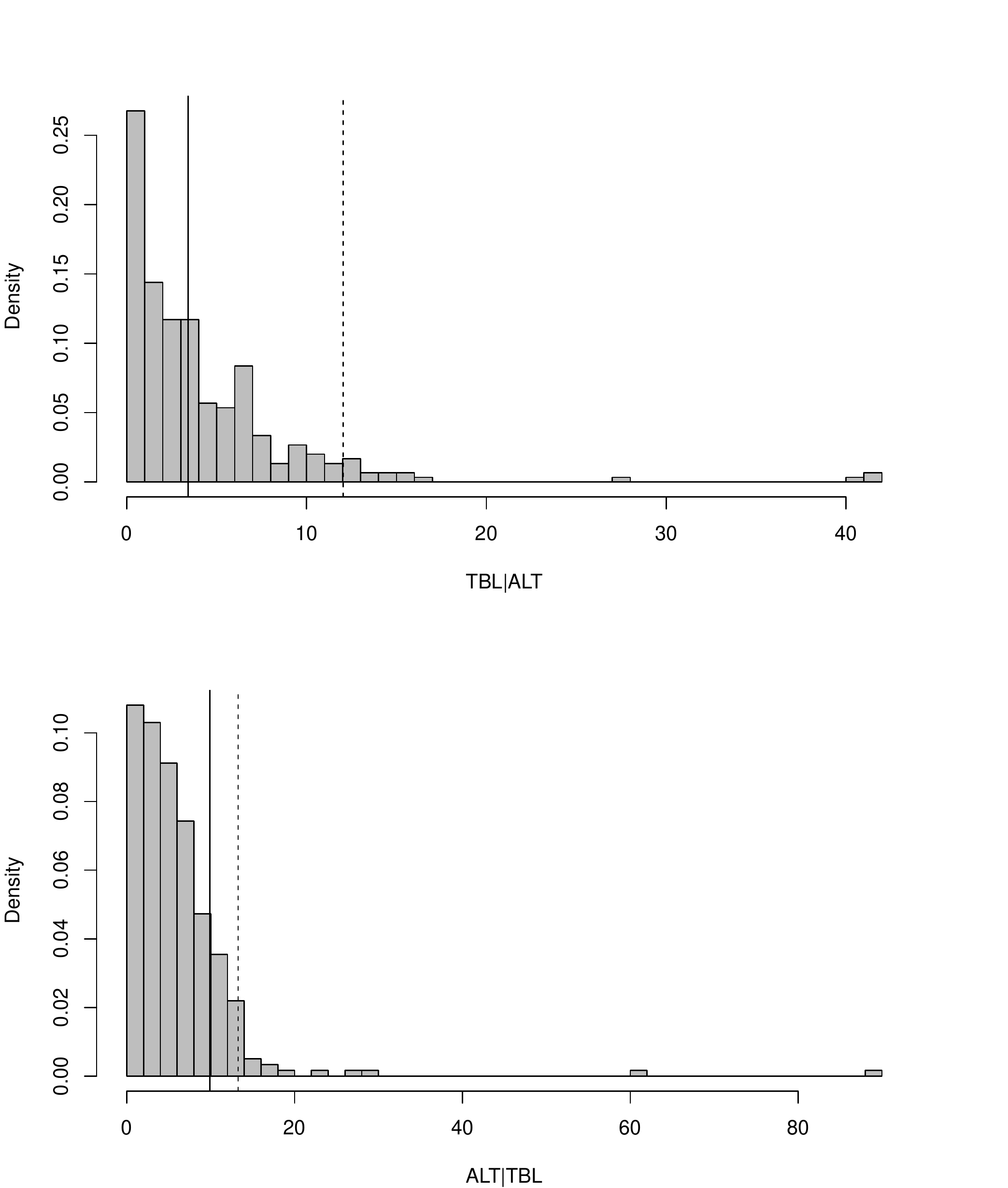}
\caption{Simulated distribution of the likelihood ratio test statistic
  described in Section~\ref{sec:inferencestoch} for testing the
  hypothesis of ordering between the four doses for the conditional
  distributions of $Y_2 \mid Y_1 > v$ (top) and $Y_1 \mid Y_2 > v$
  (bottom), where $v$ is the 99.7\% quantile of the standard Laplace
  distribution. The solid and dashed vertical lines correspond to the
  observed statistic and the 95\% quantile of the simulated
  distribution, respectively.}
\label{fig:lrt}
\end{figure}
\clearpage
%%---------------------------------------------------------------------
%\subsection{Conditional quantile estimates}
%\label{sec:Condquant}
%%---------------------------------------------------------------------
\begin{figure}[h!]
  \centering
  \includegraphics[scale=.8,trim= 30 40 30 30]{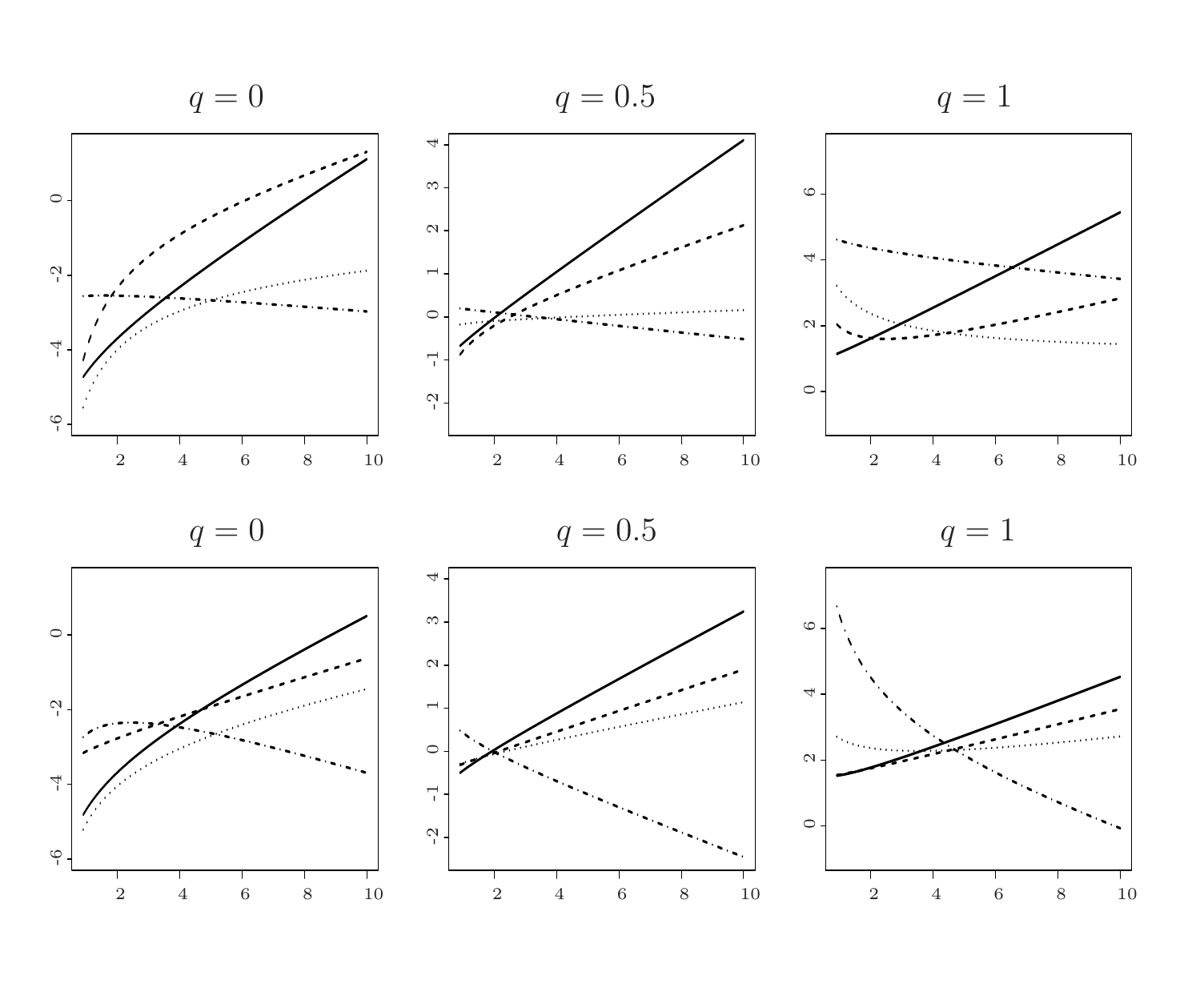}
  % \pstool[width=6in,height=5in]{condquantiles2.eps} {
  %  \psfrag{x}{} \psfrag{QQQQQQQ}{}
  %  \psfrag{-6}{\tiny-$6$} \psfrag{-5}{\tiny-$5$}
  %  \psfrag{-4}{\tiny-$4$} \psfrag{-3}{\tiny-$3$}
  %  \psfrag{-2}{\tiny-$2$} \psfrag{-1}{\tiny-$1$}
  %  \psfrag{12}{\tiny$12$} \psfrag{14}{\tiny$14$}
  %  \psfrag{0}{\tiny$0$} \psfrag{1}{\tiny$1$} \psfrag{6}{\tiny$6$}
  %  \psfrag{5}{\tiny$5$} \psfrag{4}{\tiny$4$} \psfrag{3}{\tiny$3$}
  %  \psfrag{7}{\tiny$7$} \psfrag{8}{\tiny$8$} \psfrag{2}{\tiny$2$}
  %  \psfrag{0.00}{\tiny$0.0$} \psfrag{0.05}{\tiny$0.05$}
  %  \psfrag{0.10}{\tiny$0.1$} \psfrag{0.15}{\tiny$0.15$}
  %  \psfrag{0.1}{\tiny$0.1$} \psfrag{0.5}{\tiny$0.5$}
  %  \psfrag{0.2}{\tiny$0.2$} \psfrag{1.0}{\tiny$1.0$}
  %  \psfrag{0.3}{\tiny$0.3$} \psfrag{1.5}{\tiny$1.5$}
  %  \psfrag{0.4}{\tiny$0.4$} \psfrag{0.20}{\tiny$0.2$}
  %  \psfrag{0.25}{\tiny$0.25$} \psfrag{0.30}{\tiny$0.3$}
  %  \psfrag{0.000}{\tiny$0.0$} \psfrag{0.005}{\tiny$0.005$}
  %  \psfrag{0.010}{\tiny$0.01$} \psfrag{0.015}{\tiny$0.015$}
  %  \psfrag{10}{\tiny$10$} \psfrag{15}{\tiny$15$}
  %  \psfrag{20}{\tiny$20$} \psfrag{0.0}{\scriptsize$0.0$}
  %  \psfrag{AAAAAA}{\hspace{-6pt}$q=0$} \psfrag{BBBBBB}{\hspace{-10pt}$q=0.5$}
  %  \psfrag{CCCCCC}{\hspace{-8pt}$q=1$} \psfrag{DDDDDD}{\hspace{-6pt}$q=0$}
  %  \psfrag{EEEEEE}{\hspace{-10pt}$q=0.5$} \psfrag{FFFFFF}{\hspace{-8pt}$q=1$} }
   \caption{Conditional quantile estimates $y_{2|1}(q)$ of
     $Y_{2}^{j}|Y_{1}^{j}=x$ under HT (first row) and SO models
     (second row). The line types correspond to {- $\cdot$ - $\cdot$
       -} dose $A$,\
     {\color{black}\leavevmode\hdashrule[0.1cm]{1.2cm}{0.022cm}{0.3mm
         1mm}} dose $B$,\
     {\color{black}\leavevmode\hdashrule[0.1cm]{1.2cm}{0.02cm}{1mm
         0.5mm}} dose $C$,\
     {\color{black}\leavevmode\hdashrule[0.1cm]{1.2cm}{0.02cm}{3mm
         0mm}} dose $D$. First, second and third column show the
     minimum ($q=0$), median ($q=0.5$) and maximum ($q=1$) conditional
     quantiles, respectively.}
   \label{fig:condquantiles}
 \end{figure}

% \subsection{Prediction}
% \label{sec:prediction} In this section, the main
\begin{figure}[!htbp]
  \centering
  \includegraphics[scale=1.,trim= 10 10 10 10]{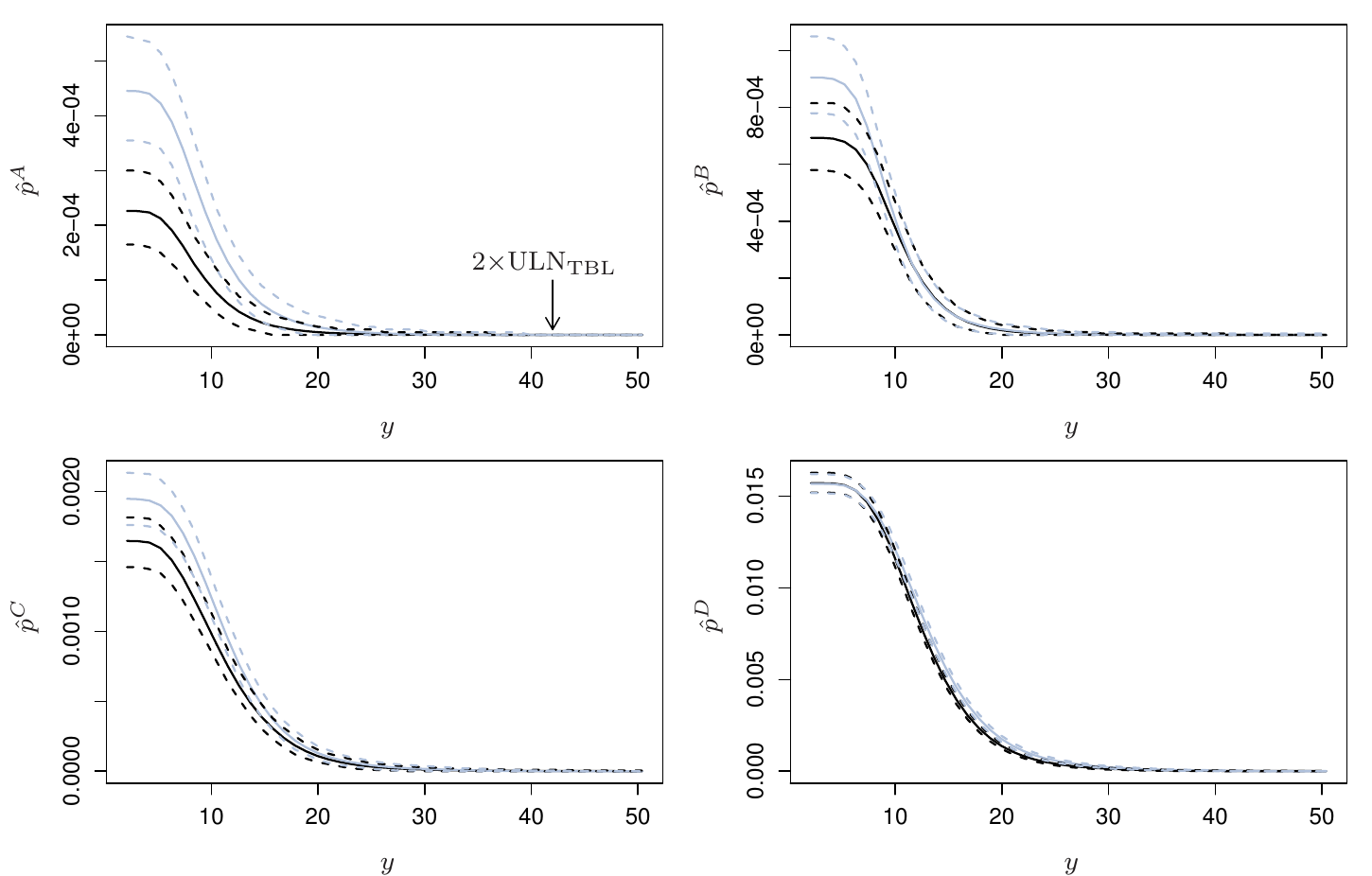}
  % \pstool[scale=0.58]{predictions.eps} 
  % {
  %   \psfrag{2xULN}{\scriptsize \hspace{-16pt}2$\times \text{ULN}_{\text{TBL}}$}
  %   \psfrag{PP}{\scriptsize $\hat{p}^{A}$}
  %   \psfrag{PB}{\scriptsize $\hat{p}^{B}$}
  %   \psfrag{PC}{\scriptsize $\hat{p}^{C}$}
  %   \psfrag{PD}{\scriptsize $\hat{p}^{D}$}
  %   \psfrag{y}{\scriptsize $y$}
  % }
  \caption{Estimated survival probabilities $\hat{p}^{j}(x,y)$ for all
    doses with $x=3\times\text{ULN}_{\text{ALT}}$ and variable
    $y$. Solid and dashed lines correspond the point-wise estimates
    and their $95\%$ confidence intervals, respectively. Black and
    grey colour shows estimates from the SO model and HT models,
    respectively.}
  \label{fig:predictions}
\end{figure}

\newpage
\section{Discussion}
As identified by \cite{southheff12,southheff12b}, liver toxicity can
be assessed by the joint extremes of ALT and TBL. However, due to the
limited sample size and the insufficient duration of the clinical
trial (6 weeks only), extrapolation to the tail area that identifies
DILI is not feasible for the laboratory data that have been analysed
in this paper.\ \cite{southheff12b} found some dose response
relationship for the probability of joint extreme elevations but
attributed this pattern to the large number of cases with
$\text{ALT}>3\times \text{ULN}_{\text{ALT}}$ in the higher dose groups
rather than an effect on TBL or stronger extremal dependence. Here, we
have developed methodology for ordered tail dependence across doses, a
pattern that is potentially triggered by toxicity but not formally
assessed by \cite{southheff12b}. Based on current biological
understanding, we view this pattern as an alternative measure of
altered liver behaviour and our aim in this analysis is to formally
test ordered dependence in the joint tail area of baseline-adjusted
ALT and TBL.

Our model formulation builds on \cite{southheff12b} model and extends
the \cite{hefftawn04} conditional approach, to account for stochastic
ordering in the tails for assessing DILI in multiple dose trials. Our
approach consists of bounding conditional distribution functions
through additional constraints on the parameter space of
\cite{hefftawn04} model. These constraints are used to construct
likelihood ratio tests which allow model selection and potential
efficiency gains in estimation as shown mainly by our simulations for
asymptotically independent models.

% We have extended the conditional extremes model of \cite{hefftawn04}
% and \cite{keefpaptawn11} to account for ordered tail dependence and
% used it to analyse liver-related laboratory safety data that have been
% linked with toxicity \citep{southheff12b}.\ 

% Our motivation for examining this measure stems from the fact that
% liver cells leak ALT into the blood as they are damaged. As the
% amount of damage increases, the amount of ALT increases, and so the
% liver begins to lose its capacity to clear TBL. Subsequently, TBL is
% also expected to start increase. At levels of damage that do not
% affect liver's ability to clear TBL, dependence is not
% expected. Hence, given that the drug has increasing toxicity with
% dose, we expect a natural ordering in the joint tail area of ALT and
% TBL.

Our main finding that complements \cite{southheff12b} analysis is
statistical evidence of ordered tail dependence across doses which we
view as a signal of altered liver behaviour.\ Our results and
conclusions predict slightly higher probabilities of extreme
elevations than those predicted originally by \cite{southheff12b} but
of the same order of magnitude. This is possibly a consequence of the
modelling of baseline variables which allows extrapolation in the
marginal tails but could also be attributed to the different robust
regression approach used here to adjust the baseline effect. Also, the
predicted survival curves indicate ordering from both unconstrained
and constrained modelling approaches. This feature stems primarily
from the conditional dependence model estimates of baseline adjusted
ALT and TBL which show ordering for a range of
quantiles.% than those for doses that have not been linked with
% toxicity and similar findings with for the toxic dose.

Last, there are some caveats with the proposed ordering effect used as
a measure of altered liver behaviour, especially when considering
highly toxic drugs for prolonged periods. If much damage has been done
so that there is no ALT left to leak into the blood, we would expect
ALT to come back down but TBL to remain high. The proposed methodology
though could still be used to monitor such patterns in longitudinal
trials via tests of dose ordering at consecutive time points.

\subsection*{Acknowledgments}
Ioannis Papastathopoulos acknowledges funding from AstraZeneca and the
SuSTaIn program -~EPSRC~grant EP/D063485/1 - at the School of
Mathematics, University of Bristol.\ We would particularly like to
thank Harry Southworth of AstraZeneca, two Referees and the Associate
Editor for helpful discussions and constructive comments on the
analysis of the pharmaceutical data.

% ------------------------------------------------------------------------
% References and BiBTeX style.
% ------------------------------------------------------------------------

\bibliographystyle{agsm} %chicago.sty

% \bibliography{/home/ip12483/Brunel_Fellowship/Bibliography/BIBTeX/bibliography.bib}

% ------------------------------------------------------------------------
% Values of parameters where D has 2 t.p.  x <- seq(.01,10,len=100000)
% zj <- 0.6; zk <- 0.6 c <- 0.1; b <- 0.2 ; d <- 0.5
% ------------------------------------------------------------------------

\end{document}